\documentclass[11pt]{article}

\usepackage[margin=1in]{geometry}
\usepackage{amsthm,amssymb,amsmath}  
\usepackage{xspace,enumerate}
\usepackage[dvipsnames]{xcolor}
\usepackage[colorlinks=true,urlcolor=Blue,citecolor=Green,linkcolor=BrickRed]{hyperref}
\usepackage[capitalise]{cleveref}
\usepackage[utf8]{inputenc}
\usepackage{thmtools}
\usepackage{thm-restate}
\usepackage{authblk}
\usepackage{graphics,adjustbox}

\theoremstyle{plain}
\newtheorem{theorem}{Theorem}

\newtheorem{fact}{Fact}
\newtheorem{observation}{Observation}
\theoremstyle{definition}

\newtheorem{example}{Example}

\newcommand{\cO}{\mathcal{O}}
\newcommand{\LCPF}{$\textsf{LPCF}_{k'}$\xspace}
\newcommand{\LPal}{$\textsf{LPALCF}$\xspace}

\newcommand{\sqms}{\textsf{SQMS}}
\newcommand{\per}{\textsf{per}}
\newcommand{\ST}{\textsf{ST}}
\newcommand{\GST}{\textsf{GST}\xspace}
\newcommand{\aGST}{\textsf{aGST}\xspace}
\newcommand{\MS}{\textsf{MS}}
\newcommand{\ie}{\textit{i.e., }}
\makeatletter
\def\smallunderbrace#1{\mathop{\vtop{\m@th\ialign{##\crcr
   $\hfil\displaystyle{#1}\hfil$\crcr
   \noalign{\kern3\p@\nointerlineskip}%
   \tiny\upbracefill\crcr\noalign{\kern3\p@}}}}\limits}
\makeatother

\makeatletter
\def\mathcolor#1#{\@mathcolor{#1}}
\def\@mathcolor#1#2#3{%
  \protect\leavevmode
  \begingroup
    \color#1{#2}#3%
  \endgroup
}
\makeatother

\definecolor{darkgreen}{rgb}{0, .6, 0}
\usepackage{microtype}

\title{Longest Property-Preserved Common Factor}
\author[1]{Lorraine A.K Ayad}
\author[2]{Giulia Bernardini}
\author[3]{Roberto Grossi}
\author[4]{Costas S. Iliopoulos}
\author[5]{Nadia Pisanti}
\author[6]{Solon P. Pissis}
\author[7]{Giovanna Rosone}

\affil[1]{Department of Informatics, King's College London, London, UK\\
\texttt{lorraine.ayad@kcl.ac.uk}}
\affil[2]{Department of Informatics, Systems and Communication (DISCo), University of Milan-Bicocca, Italy\\
\texttt{giulia.bernardini@unimib.it}}
\affil[3]{Department of Computer Science, University of Pisa, Italy and ERABLE Team, INRIA, France\\
\texttt{grossi@di.unipi.it}}
\affil[4]{Department of Informatics, King's College London, London, UK\\ 
\texttt{c.iliopoulos@kcl.ac.uk}}
\affil[5]{Department of Computer Science, University of Pisa, Italy and ERABLE Team, INRIA, France\\
\texttt{pisanti@di.unipi.it}}
\affil[6]{Department of Informatics, King's College London, London, UK\\ 
\texttt{solon.pissis@kcl.ac.uk}}
\affil[7]{Department of Computer Science, University of Pisa, Italy\\
\texttt{giovanna.rosone@unipi.it}}

\begin{document}
\date{}
\maketitle
\begin{abstract}
In this paper we introduce a new family of string processing problems.
We are given two or more strings and we are asked to compute a factor common to all strings that preserves a specific property and has maximal length. Here we consider three fundamental string properties: square-free factors, periodic factors, and palindromic factors under three different settings, one per property. In the first setting, we are given a string $x$ and we are asked to construct a data structure over $x$ answering the following type of on-line queries: given string $y$, find a longest square-free factor common to $x$ and $y$. In the second setting, we are given $k$ strings and an integer $1 < k'\leq k$ and we are asked to find a longest periodic factor common to at least $k'$ strings. In the third setting, we are given two strings and we are asked to find a longest palindromic factor common to the two strings.
We present linear-time solutions for all settings. We anticipate that our paradigm can be extended to other string properties or settings. 
\end{abstract}

\section{Introduction}

In the longest common factor problem, also known as longest common substring problem, we are given two strings $x$ and $y$, each of length at most $n$, and we are asked to find a maximal-length string occurring in both $x$ and $y$. This is a classical and well-studied problem in computer science arising out of different practical scenarios. It can be solved in $\cO(n)$ time and space~\cite{ChiHui1992,DBLP:books/cu/Gusfield1997} (see also~\cite{DBLP:conf/esa/KociumakaSV14,DBLP:conf/cpm/StarikovskayaV13}). Recently, the same problem has been extensively studied under distance metrics; that is, the sought factors (one from $x$ and one from $y$) must be at distance at most $k$ and have maximal length~\cite{DBLP:conf/cpm/Charalampopoulos18,DBLP:journals/jcb/ThankachanAA16,DBLP:conf/recomb/ThankachanACA18,DBLP:journals/corr/abs-1801-04425,MaxHamming,jda08} (and references therein).

In this paper we initiate a new related line of research. We are given two or more strings and our goal is to compute a {\em factor} common to all strings that preserves a specific {\em property} and has maximal length. An analogous line of research was introduced in~\cite{Chowdhury}. It focuses on computing a {\em subsequence} (rather than a factor) common to all strings that preserves a specific property and has maximal length. Specifically, in~\cite{Chowdhury,BAE,INENAGA}, the authors considered computing a longest common palindromic subsequence and in~\cite{inoue_et_al} computing a longest common square subsequence.

We consider three fundamental string properties: {\em square-free} factors, {\em periodic}, and {\em palindromic} factors~\cite{lothaire_2005} under three different settings, one per property. In the first setting, we are given a string $x$ and we are asked to construct a data structure over $x$ answering the following type of on-line queries: given string $y$, find a longest square-free factor common to $x$ and $y$. In the second setting, we are given $k$ strings and an integer $1 < k'\leq k$ and we are asked to find a longest periodic factor common to at least $k'$ strings. 
In the third setting, we are given two strings and we are asked to find a longest palindromic factor common to the two strings.
We present linear-time solutions for all settings. We anticipate that our paradigm can be extended to other string properties or settings.

\subsection{Definitions and Notation}

An \textit{alphabet} $\Sigma$ is a non-empty finite ordered set of letters of size $\sigma=|\Sigma|$. 
In this work we consider that $\sigma = \cO(1)$ or that $\Sigma$ is a linearly-sortable integer alphabet. 
A \textit{string} $x$ on an alphabet $\Sigma$ is a sequence of elements of $\Sigma$. The set of all strings on an alphabet $\Sigma$, including the \textit{empty string} $\varepsilon$ of length $0$, is denoted by $\Sigma^*$. For any string $x$, we denote by $x[i..j]$ the \textit{substring} (sometimes called \textit{factor}) of $x$ that starts at position $i$ and ends at position $j$. In particular, $x[0 .. j]$ is the \textit{prefix} of $x$ that ends at position $j$, and $x[i..|x|-1]$ is the \textit{suffix} of $x$ that starts at position $i$, where $|x|$ denotes the \textit{length} of $x$. A string $uu$, $u \in \Sigma^*$, is called a {\em square}. A {\em square-free} string is a string that does not contain a square as a factor.

A \emph{period} of $x[0..|x|-1]$ is a positive integer $p$ such that $x[i]=x[i+p]$ holds for all $0 \leq i < |x|-p$. The smallest period of $x$ is denoted by $\per(x)$.
String $u$ is called {\em periodic} if and only if $\per(u) \leq |u|/2$.
A \emph{run} of string $x$ is an interval $[i,j]$ such that for the smallest period $p=\per(x[i..j])$ it holds that $2p \leq j-i+1$ and the periodicity cannot be extended to the left or right, \ie $i=0$ or $x[i-1] \neq x[i+p-1]$, and, $j=|x|-1$ or $x[j-p+1] \neq x[j+1]$. 

We denote the {\em reversal} of $x$ by string $x^R$, i.e. $x^R=x[|x|-1]x[|x|-2] \ldots x[0]$. 
A string $p$ is said to be a \emph{palindrome} if and only if $p=p^R$. If factor $x[i..j]$, $0 \leq i \leq j \leq n-1$, of string $x$ of length $n$ is a palindrome, then $\frac{i+j}{2}$ 
is the \emph{center} of $x[i..j]$ in $x$ and $\frac{j-i+1}{2}$ is the {\em radius} of $x[i.. j]$. 
In other words, a palindrome is a string that reads the same forward and backward, i.e. a string $p$ is a palindrome if $p = y a y^R$ where $y$ is a string, $y^R$ is the reversal of $y$ 
and $a$ is either a single letter or the empty string. Moreover, $x[i.. j]$ is called a {\em palindromic factor} of $x$. 
It is said to be a \emph{maximal palindrome} if there is no other palindrome in $x$ with center $\frac{i+j}{2}$ and larger radius. 
Hence $x$ has exactly $2n-1$ maximal palindromes. A maximal palindrome $p$ of $x$ can be encoded as a pair $(c,r)$, where $c$ is the center of $p$ in $x$ and $r$ is the radius of $p$.

\subsection{Algorithmic Toolbox}

The maximum number of runs in a string of length $n$ is less than $n$~\cite{BannaiTomohiroInenagaNakashimaTakedaTsuruta2017}, and, moreover, all runs can be computed in $\cO(n)$ time~\cite{KolpakovKucherov1999,BannaiTomohiroInenagaNakashimaTakedaTsuruta2017}.

The \textit{suffix tree} $\ST(x)$ of a non-empty string $x$ of length $n$ is a compact trie representing all suffixes of $x$. $\ST(x)$ can be constructed in $\cO(n)$ time~\cite{farach1997optimal}. We can analogously define and construct the \textit{generalised suffix tree} $\GST(x_0,x_1,\ldots,x_{k-1})$ for a set of $k$ strings. We assume the reader is familiar with these data structures.

The matching statistics capture all matches between two strings $x$ and $y$~\cite{Chang1994}. 
More formally, the {\em matching statistics} of a string $y[0..|y|-1]$ with respect to a string $x$ is an array $\MS_y[0..|y|-1]$, where $\MS_y[i]$ is a pair $(\ell_i, p_i)$ such that
(i) $y[i..i + \ell_i-1]$ is the longest prefix of $y[i..|y|-1]$ that is a factor of $x$; and (ii) $x[p_i..p_i + \ell_i-1] = y[i..i + \ell_i-1]$. Matching statistics can be computed in $\cO(|y|)$ time for $\sigma = \cO(1)$ by using $\ST(x)$ \cite{DBLP:books/cu/Gusfield1997,DBLP:conf/spire/BelazzouguiC14,tcs09}.

Given a rooted tree $T$ with $n$ leaves coloured from $0$ to $k-1$, $1 < k \leq n$, the \emph{colour set size} problem is finding, for each internal node $u$ of $T$, the number of different leaf colours in the subtree rooted at $u$. In \cite{ChiHui1992}, the authors present an $\cO(n)$-time solution to this problem. 

In the \textit{weighted ancestor} problem, introduced in~\cite{DBLP:conf/cpm/FarachM96}, we consider a rooted tree $T$ with an integer weight function $\mu$ defined on the nodes. We require that the weight of the root is zero and the weight of any other node is strictly larger than the weight of its parent. A weighted ancestor query, given a node $v$ and an integer value $\ell\le \mu(v)$, asks for the highest ancestor $u$ of $v$ such that $\mu(u)\ge \ell$, \ie such an ancestor $u$ that $\mu(u)\ge \ell$ and $\mu(u)$ is the smallest possible. When $T$ is the suffix tree of a string $x$ of length $n$, we can locate the locus of any factor of $x[i..j]$ using a weighted ancestor query. We define the weight of a node of the suffix tree as the length of the string it represents. Thus a weighted ancestor query can be used for the terminal node corresponding to $x[i..n-1]$ to create (if necessary) and mark the node that corresponds to $x[i..j]$. Given a collection $Q$ of weighted ancestor queries on a weighted tree $T$ on $n$ nodes with integer weights up to $n^{\cO(1)}$, all the queries in $Q$ can be answered {\em off-line} in $\cO(n+|Q|)$ time~\cite{DBLP:journals/corr/BartonKLPR17}.
    
\section{Square-Free-Preserved Matching Statistics}

In this section, we introduce the square-free-preserved matching statistics problem and provide a linear-time solution. In the {\em square-free-preserved matching statistics} problem we are given a string $x$ of length $n$ and we are asked to construct a data structure over $x$ answering the following type of on-line queries: given string $y$, find the longest square-free prefix of $y[i..|y|-1]$ that is a factor of $x$, for all $0 \leq i < |y|-1$. 
(For related work see~\cite{DBLP:conf/spire/DumitranMN15}.)
We represent the answer using an integer array $\sqms_y[0..|y|-1]$ of lengths, but we can trivially modify our algorithm to report the actual factors. It should be clear that a maximum element in $\sqms$ gives the length of some longest square-free factor common to $x$ and $y$. 

\emph{Construction.} Our data structure over string $x$ consists of the following:
\begin{itemize}
\item An integer array $L_x[0..n-1]$, where $L_x[i]$ stores the length of the longest square-free factor starting at position $i$ of string $x$.
\item The suffix tree $\ST(x)$ of string $x$.
\end{itemize}

The idea for constructing array $L_x$ efficiently is based on the following crucial observation. 

\begin{observation}
\label{obs:sq}
If $x[i..n-1]$ contains a square then $L_x[i]+1$, for all $0 \leq i < n$, is the length of the {\em shortest prefix} of $x[i..n-1]$ (factor $f$) containing a square. In fact, the square is a suffix of $f$, otherwise $f$ would not have been the shortest.
If $x[i..n-1]$ does not contain a square then $L_x[i]=n-i$.
\end{observation}

We thus shift our focus to computing the shortest such prefixes.
We start by considering the runs of $x$. Specifically, we consider squares in $x$ observing that a run $[\ell,r]$ with period $p$ contains $r-\ell - 2p + 2$ squares of length $2p$ with the leftmost one starting at position $\ell$. Let $r'\!=\!\ell\!+\!2p\!-\!1$ denote the ending position of the leftmost such square of the run. In order to find, for all $i$'s, the shortest prefix of $x[i..n-1]$ containing a square $s$, and thus compute $L_x[i]$, we have two cases:

\begin{enumerate}
\item \label{item:case1}
$s$ is part of a run $[\ell,r]$ in $x$ that starts {\em after} $i$. In particular, $s=x[\ell..r']$ such that $r' \leq r$, $\ell>i$, and $r'$ is minimal. In this case the shortest factor has length $\ell+2p-i$; we store this value in an integer array $C[0..n-1]$. If no run starts after position $i$ we set $C[i]=\infty$. 
To compute $C$, after computing in $\cO(n)$ time all the runs of $x$ with their $p$ and $r'$~\cite{KolpakovKucherov1999,BannaiTomohiroInenagaNakashimaTakedaTsuruta2017}, we sort them by $r'$. A right-to-left scan after this sorting associates to $i$ the closest $r'$ with $\ell>i$.

\item \label{item:case2}
$s$ is part of a run $[\ell,r]$ in $x$ and $i\!\in\![\ell,r]$. This implies that if $i\!\leq\!r\!-\!2p\!+\!1$ then a square {\em starts at} $i$ and we store the length of the shortest such square in an integer array $S[0..n-1]$. If no square starts at position $i$ we set $S[i]=\infty$. Array $S$ can be constructed in $\cO(n)$ time by applying the algorithm of~\cite{DUVAL2004229}.
\end{enumerate}

Since we do not know which of the two cases holds, we compute both $C$ and $S$. By Observation~\ref{obs:sq}, if $C[i]=S[i]=\infty$ ($x[i..n-1]$ does not contain a square) we set $L_x[i] =n-i$; otherwise ($x[i..n-1]$ contains a square) we set $L_x[i]=\min\{C[i],S[i]\}-1$. 

Finally, we build the suffix tree $\ST(x)$ of string $x$ in $\cO(n)$ time~\cite{farach1997optimal}. This completes our construction.

\emph{Querying.} We rely on the following fact for answering the queries efficiently. 
\begin{fact}
\label{fct:sqf}
Every factor of a square-free string is square-free.
\end{fact}

Let string $y$ be an on-line query. Using $\ST(x)$, we compute the matching statistics $\MS_y$ of $y$ with respect to $x$. For each $j \in [0,|y|-1]$, $\MS_y[j]=(\ell_i,i)$ indicates that $x[i..i+\ell_i-1] = y[j..j+\ell_i-1]$. This computation can be done in $\cO(|y|)$ time~\cite{DBLP:books/cu/Gusfield1997,DBLP:conf/spire/BelazzouguiC14}. By applying Fact~\ref{fct:sqf}, we can answer any query $y$ in $\cO(|y|)$ time for $\sigma = \cO(1)$ by setting $\sqms_y[j]=\min\{\ell_i, L_x[i]\}$, for all $0 \leq j \leq |y| - 1$.

We arrive at the following result.

\begin{theorem}\label{th:sqms}
Given a string $x$ of length $n$ over an alphabet of size $\sigma = \cO(1)$, we can construct a data structure of size $\cO(n)$ in time $\cO(n)$, answering $\sqms_y$ on-line queries in $\cO(|y|)$ time.
\end{theorem}

\begin{proof}
The time complexity of our algorithm follows from the above discussion.

We next show the correctness of our algorithm. Let us first show the correctness of computing array $L_x$.
The square contained in the shortest prefix of $x[i..n-1]$ (containing a square) starts by definition either at $i$ or after $i$. If it starts at $i$ this is correctly computed by the algorithm of~\cite{DUVAL2004229} which assigns the length of the shortest such square in $S[i]$. If it starts after $i$ it must be the leftmost square of another run by the runs definition. $C[i]$ stores the length of the shortest prefix containing such a square. Then by Observation~\ref{obs:sq}, $L_x[i]$ is computed correctly. 

It suffices to show that, if $w$ is the longest square-free substring common to $x$ and $y$ occurring at position $i_x$ in $x$ and at position $i_y$ in $y$, then 
(i) $\MS_y[i_y]=(\ell,i_x)$ with $\ell\geq |w|$ and $x[i_x..i_x+\ell-1]=y[i_y..i_y+\ell-1]$; 
(ii) $w$ is a prefix of $x[i_x..i_x+L_x[i_x]-1]$; and 
(iii) $\sqms_y[i_y]=|w|$. 
Case (i) directly follows from the correctness of the matching statistics algorithm.
For Case (ii), since $w$ occurs at $i_x$ and $w$ is square-free, $L_x[i_x] \geq |w|$.
For Case (iii), since $w$ is square-free we have to show that $|w| = \min\{\ell_i, L_x[i]\}$.
We know from (i) that $\ell \geq |w|$ and from (ii) that $L_x[i_x]\geq |w|$.
If $\min\{\ell_i, L_x[i]\}=\ell$, then $w$ cannot be extended because the possibly longer than $|w|$ square-free string occurring at $i_x$ does not occur in $y$, and in this case $|w| =\ell$. Otherwise, if 
$\min\{\ell_i, L_x[i]\}=L_x[i_x]$ then $w$ cannot be extended because it is no longer square-free, and in this case $|w| = L_x[i_x]$. Hence we conclude that $\sqms_y[i_y]=|w|$. The statement follows. 
\end{proof}

The following example provides a complete overview of the workings of our algorithm.

\begin{example}
Let $x = \texttt{aababaababb}$ and $y = \texttt{babababbaaab}$. The length of a longest common square-free factor is 3, and the factors are $\texttt{bab}$ and $\texttt{aba}$.

\centering
\begin{tabular}{ccccccccccccccc} \\\cline{1-12}
$i$ & 0 & 1 & 2 & 3 & 4 & 5 & 6 & 7 & 8 & 9 & 10 \\ \cline{1-12}
$x[i]$ & \texttt{a} & \texttt{a} & \texttt{b} & \texttt{a} & \texttt{b} & \texttt{a} & \texttt{a} & \texttt{b} & \texttt{a} & \texttt{b} & \texttt{b} \\
$C[i]$ & 5 & 6 & 5 & 4 & 3 & 5 & 5 & 4 & 3 & $\infty$ & $\infty$\\
$S[i]$ & 2 & 4 & 4 & 6 & $\infty$ & 2 & 4 & $\infty$ & $\infty$ & 2 & $\infty$ \\
$L_x[i]$ & 1 & 3 & 3 & 3 & 2 & 1 & 3 & 3 & 2 & 1 & 1 \\ \hline
$j$ & $0$ & $1$ & $2$ & $3$ & $4$ & $5$ & $6$ & $7$ & $8$ & $9$ & $10$ & $11$  \\  \hline
$y[j]$ & \texttt{b} & \texttt{a} & \texttt{b} & \texttt{a} & \texttt{b} & \texttt{a} & \texttt{b} & \texttt{b} & \texttt{a} & \texttt{a} & \texttt{a} & \texttt{b}\\
$\MS_y[j]$ & (4,2) & (5,1) & (4,2) & (5,6) & (4,7) & (3,8) &(2,9) & (3,4) & (2,0) & (3,0) & (2,1) & (1,2) \\
$\sqms_y[j]$ & 3 & 3 & 3 & 3 & 3 & 2 & 1 & 2 & 1 & 1 & 2 & 1 
\end{tabular}
\end{example}

\section{Longest Periodic-Preserved Common Factor}

In this section, we introduce the longest periodic-preserved common factor problem and provide a linear-time solution. In the {\em longest periodic-preserved common factor} problem, we are given $k \ge 2$ strings $x_0,x_1,\dots, x_{k-1}$ of total length $N$ and an integer $1 < k'\le k$, and we are asked to find a longest periodic factor common to at least $k'$ strings. In what follows we present two different algorithms to solve this problem. We represent the answer \LCPF by the length of a longest factor, but we can trivially modify our algorithms to report an actual factor. Our first algorithm, denoted by {\sc lPcf}, works as follows.

\begin{enumerate}
\item Compute the runs of string $x_j$, for all $0 \le j < k$.
\item Construct the generalised suffix tree $\GST(x_0,x_1,\ldots,x_{k-1})$ of $x_0,x_1,\ldots,x_{k-1}$. 
\item For each string $x_j$ and for each run $[\ell,r]$ with period $p_\ell$ of $x_j$, augment \GST with the explicit node spelling $x_j[\ell..r]$, decorate it with $p_\ell$, and mark it as a \emph{candidate} node. This can be done as follows: for each run $[\ell,r]$ of $x_j$, for all $0\!\le\!j\!<\!k$, find the leaf corresponding to $x_{j}[\ell..|x_j|\!\!-\!\!1] $ and answer the weighted ancestor query in \GST with weight $r\!-\!\ell\!+\!1$. Moreover, mark as candidates all {\em explicit} nodes spelling a prefix of length $d$ of any run $[\ell,r]$ with $2p_\ell\leq d$. 
\item Mark as \emph{good} the nodes of the tree having at least $k'$ different colours on the leaves of the subtree rooted there. Let \aGST be this augmented tree. 
\item Return as \LCPF the string depth of a candidate node in \aGST which is also a good node, and that has maximal string depth (if any, otherwise return 0). 
\end{enumerate}
\begin{theorem}\label{the:LCPF}
Given $k$ strings of total length $N$ on alphabet $\Sigma = \{1,\ldots,N^{\cO(1)}\}$, and an integer $1 < k'\leq k$, algorithm {\sc lPcf} returns \LCPF in time $\cO(N)$.
\end{theorem}
\begin{proof}
Let us assume wlog that $k'=k$, and let $w$ with period $p$ be the longest periodic factor common to all strings. By the construction of \aGST (Steps 1-4), the path spelling $w$ leads to a good node $n_w$ as $w$ occurs in all the strings. We make the following observation. 

\begin{observation}
\label{obs:run-of-a-period}
Each periodic factor with period $p$ of string $x$ is a factor of $x[i..j]$, where $[i,j]$ is a run with period $p$. 
\end{observation}

By Observation~\ref{obs:run-of-a-period}, in all strings, $w$ is included in a run having the same period. Observe that for at least one of the strings, there is a run ending with $w$, otherwise we could extend $w$ obtaining a longer periodic common factor (similarly, for at least one of the strings, there is a run starting with $w$). Therefore $n_w$ is {\em both} a good and a candidate node. By definition, $n_w$ is at string depth at least $2p$ and, by construction, \LCPF is the string depth of a deepest such node; thus $|w|$ will be returned by Step~5. 

As for the time complexity, Step 1~\cite{KolpakovKucherov1999,BannaiTomohiroInenagaNakashimaTakedaTsuruta2017} and Step 2~\cite{farach1997optimal} can be done in $\cO(N)$ time. Since the total number of runs is less than $N$~\cite{BannaiTomohiroInenagaNakashimaTakedaTsuruta2017}, Step 3 can be done in $\cO(N)$ time using off-line weighted ancestor queries~\cite{DBLP:journals/corr/BartonKLPR17} to mark the runs as candidate nodes; and then a post-order traversal to mark their ancestor explicit nodes as candidates, if their string-depth is at least $2p_\ell$ for any run $[\ell,r]$ with period $p_\ell$. The size of the \aGST is still in $\cO(N)$. Step 4 can be done in $\cO(N)$ time~\cite{ChiHui1992}. Step 5 can be done in $\cO(N)$ by a post-order traversal of \aGST. 
\end{proof}

The following example provides a complete overview of the workings of our algorithm.

\begin{example}
Consider $x=$\texttt{ababbabba}, $y=$\texttt{ababaab}, and $k\!=\!k'\!=\!2$. The runs of $x$ are: $r_0=[0,3]$, $\per(\texttt{abab})=2$,  $r_1=[1,8]$, $\per(\texttt{babbabba})=3$, $r_2=[3,4]$, $\per(\texttt{bb})=1$, and $r_3=[6,7]$, $\per(\texttt{bb})=1$; those of $y$ are $r_4=[0,4]$, $\per(\texttt{ababa})=2$ and $r_5=[4,5]$, $\per(\texttt{aa})=1$. Fig~\ref{fig:giulia} shows \aGST for $x$, $y$, and $k\!=\!k'\!=\!2$. Algorithm {\sc lPcf} outputs $4=|\texttt{abab}|$, with $\per(\texttt{abab})=2$, as the node spelling \texttt{abab} is the deepest good one that is also a candidate.
\end{example}

\begin{figure}
 \centering
\includegraphics[width=\textwidth,height=6cm]{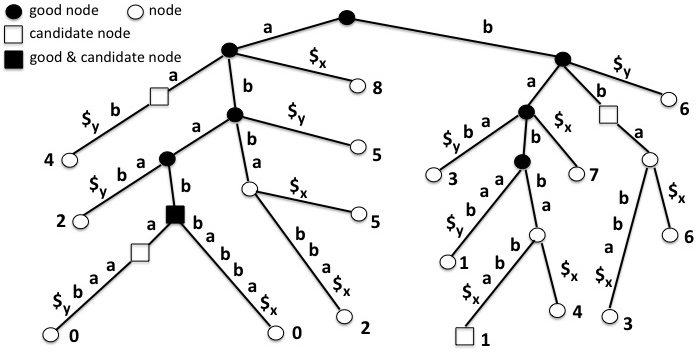}
\caption{\aGST for $x=\texttt{ababbabba}$, $y=$\texttt{ababaab}, and $k\!=\!k'\!=\!2$.}
\label{fig:giulia}
\end{figure}

We next present a second algorithm to solve this problem with the same time complexity but without the use of off-line weighted ancestor queries. The algorithm works as follows.

\begin{enumerate}
\item Compute the runs of string $x_j$, for all $0 \le j < k$.
\item Construct the generalised suffix tree $\GST(x_0,x_1,\ldots,x_{k-1})$ of $x_0,x_1,\ldots,x_{k-1}$. 
\item Mark as \emph{good} the nodes of $\GST$ having at least $k'$ different colours on the leaves of the subtree rooted there. 
\item Compute and store, for every leaf node, the \emph{nearest} ancestor that is good.
\item For each string $x_j$ and for each run $[\ell,r]$ with period $p_\ell$ of $x_j$, check the nearest good ancestor for the leaf corresponding to $x_j[\ell..|x_j|-1]$. 
Let $d$ be the string-depth of the nearest good ancestor. Then:
\begin{enumerate}
 \item If $r - \ell + 1 \leq d$, the entire run is also good.
 \item If $r - \ell + 1 > d$, check if $2p_\ell \leq d$, and if so the string for the good ancestor is periodic.
\end{enumerate}
 \item Return as \LCPF the maximal string depth found in Step 5 (if any, otherwise return 0). 
\end{enumerate}

\begin{figure}
 \centering
\includegraphics[width=\textwidth,height=9cm]{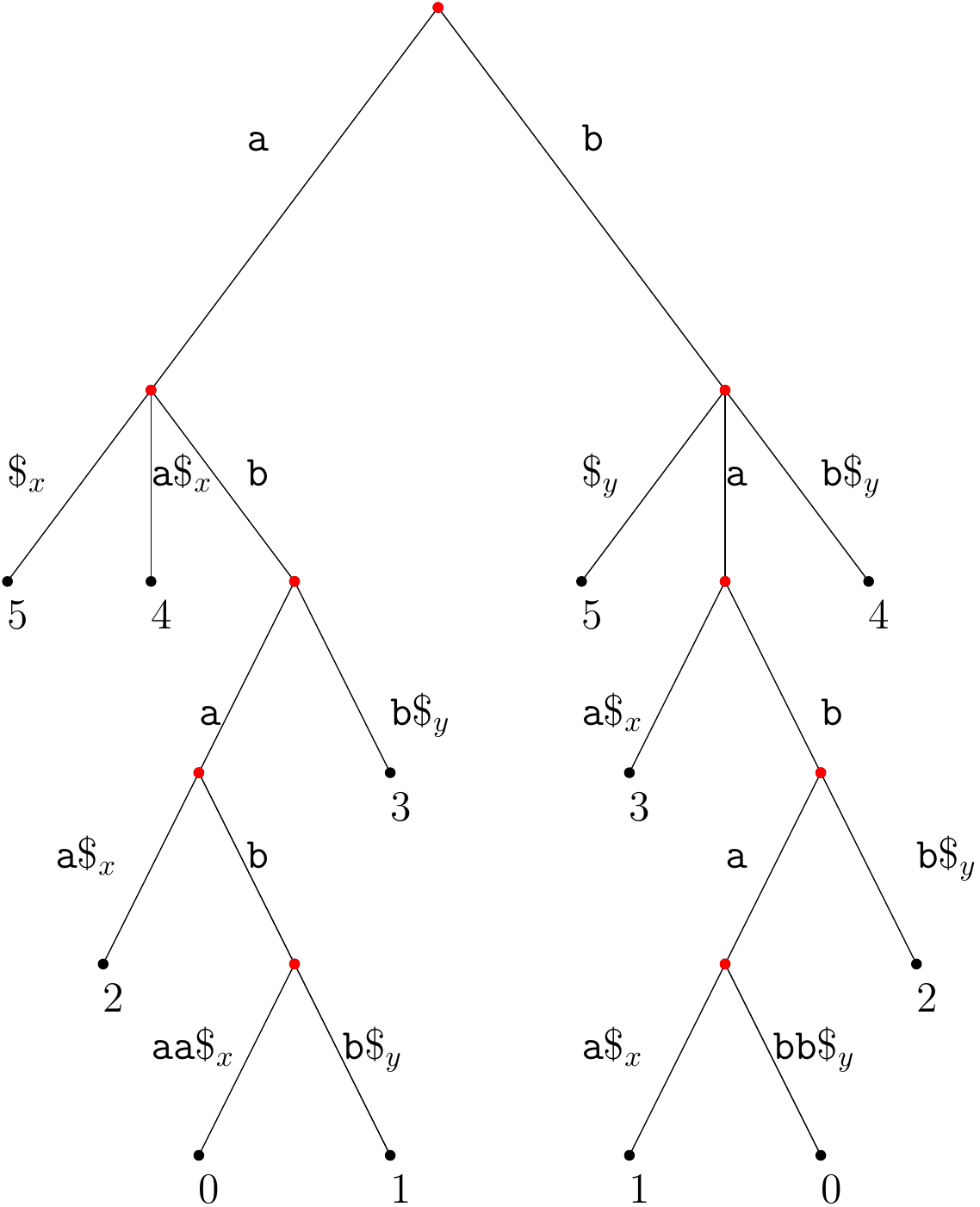}
\caption{\GST for $x=\texttt{ababaa}$, $y=\texttt{bababb}$, and $k\!=\!k'\!=\!2$. Good nodes are marked red.}
\label{fig:giulia2}
\end{figure}

Let us analyse this algorithm. Let us assume wlog that $k'=k$, and let $w$ with period $p$ be the longest periodic factor common to all strings. By the construction of $\GST$ (Steps 1-3), the path spelling $w$ leads to a good node $n_w$ as $w$ occurs in all the strings. 

By Observation~\ref{obs:run-of-a-period}, in all strings, $w$ is included in a run having the same period. 
Observe that for at least one of the strings, there is a run starting with $w$, otherwise we could extend $w$ obtaining a longer periodic common factor. 
So the algorithm should check, for each run, if there is a periodic-preserved common prefix of the run and take the longest such prefix.
\LCPF is the string depth of a deepest good node spelling a periodic factor; thus $|w|$ will be returned by Step~6. 

As for the time complexity, Step 1~\cite{KolpakovKucherov1999,BannaiTomohiroInenagaNakashimaTakedaTsuruta2017} and Step 2~\cite{farach1997optimal} can be done in $\cO(N)$ time. 
Step 3 can be done in $\cO(N)$ time~\cite{ChiHui1992} and Step 4 can be done in $\cO(N)$ time by using a tree traversal.
Since the total number of runs is less than $N$~\cite{BannaiTomohiroInenagaNakashimaTakedaTsuruta2017}, 
Step 5 can be done in $\cO(N)$ time. We thus arrive at Theorem~\ref{the:LCPF} with a different algorithm.

The following example provides a complete overview of the workings of our algorithm.

\begin{example}
Consider $x=$\texttt{ababaa}, $y=$\texttt{bababb}, and $k\!=\!k'\!=\!2$. 
The runs of $x$ are: $r_0=[0,4]$, $\per(\texttt{ababa})=2$,  $r_1=[4,5]$, $\per(\texttt{aa})=1$; those of $y$ are $r_2=[0,4]$, $\per(\texttt{babab})=2$ and $r_3=[4,5]$, $\per(\texttt{bb})=1$. 
Fig~\ref{fig:giulia2} shows $\GST$ for $x$, $y$, and $k\!=\!k'\!=\!2$. 
Consider the run $r_0=[0,4]$. The nearest good node of leaf spelling $x[0..|x|-1]$ is the node spelling $\texttt{abab}$. We have that $r - \ell + 1 = 5 > d = 4$, and $2p = 4 \leq d = 4$.
The algorithm outputs $4=|\texttt{abab}|$ as \texttt{abab} is a longest periodic-preserved common factor. Another longest periodic-preserved common factor is \texttt{baba}.
\end{example}

\section{Longest Palindromic-Preserved Common Factor}

In this section, we introduce the longest palindromic-preserved common factor problem and provide a linear-time solution. In the {\em longest palindromic-preserved common factor} problem, we are given two strings $x$ and $y$, and we are asked to find a longest palindromic factor common to the two strings. (For related work in a dynamic setting see \cite{funakoshi_et_al:LIPIcs:2018:8697,DBLP:journals/corr/abs-1804-08731}.) We represent the answer \LPal by the length of a longest factor, but we can trivially modify our algorithm to report an actual factor. Our algorithm is denoted by {\sc lPalcf}. In the description below, for clarity, we consider odd-length palindromes only. 
(Even-length palindromes can be handled in an analogous manner.)

\begin{enumerate}
\item Compute the maximal odd-length palindromes of $x$ and the maximal odd-length palindromes of $y$.
\item Collect the factors $x[i..i']$ of $x$ (resp.~the factors $y[j..j']$ of $y$) such that $i$ ($j$) is the center of an odd-length maximal palindrome of $x$ ($y$) and 
$i'$ ($j'$) is the ending position of the odd-length maximal palindrome centered at $i$ ($j$).  
\item Create a lexicographically sorted list $L$ of these strings from $x$ and $y$. 
\item Compute the longest common prefix of consecutive entries (strings) in $L$.
\item Let $\ell$ be the maximal length of longest common prefixes between any string from $x$ and any string from $y$. For odd lengths, return $\LPal = 2\ell-1$. 
\end{enumerate}

\begin{theorem}
Given two strings $x$ and $y$ on alphabet $\Sigma = \{1,\ldots,(|x|+|y|)^{\cO(1)}\}$, algorithm {\sc lPalcf} returns \LPal in time $\cO(|x|+|y|)$.
\end{theorem}
\begin{proof}

The correctness of our algorithm follows directly from the following observation.

\begin{observation}\label{obs:pal}
Any longest palindromic-preserved common factor is a factor of a maximal palindrome of $x$ with the same center and a factor of a maximal palindrome of $y$ with the same center.
\end{observation}

Step 1 can be done in $\cO(|x|+|y|)$ time~\cite{DBLP:books/cu/Gusfield1997}.
Step 2 can be done in $\cO(|x|+|y|)$ time by going through the set of maximal palindromes computed in Step 1.
Step 3 and Step 4 can be done in $\cO(|x|+|y|)$ time by constructing the data structure of~\cite{DBLP:conf/latin/Charalampopoulos18}.
Step 5 can be done in $\cO(|x|+|y|)$ time by going through the list of computed longest common prefixes.
 
\end{proof}

The following example provides a complete overview of the workings of our algorithm.

\begin{example}
Consider $x=$\texttt{ababaa} and $y=$\texttt{bababb}. 
In Step 1 we compute all maximal palindromes of $x$ and $y$.
Considering odd-length palindromes gives the following factors (Step 2) from $x$: 
$x[0..0]=\texttt{a}$, 
$x[1..2]=\texttt{\textcolor{red}{ba}}$,
$x[2..4]=\texttt{\textcolor{blue}{ab}a}$,
$x[3..4]=\texttt{\textcolor{red}{ba}}$,
$x[4..4]=\texttt{a}$, and
$x[5..5]=\texttt{a}$.
The analogous factors from $y$ are: 
$y[0..0]=\texttt{b}$, 
$y[1..2]=\texttt{\textcolor{blue}{ab}}$,
$y[2..4]=\texttt{\textcolor{red}{ba}b}$,
$y[3..4]=\texttt{\textcolor{blue}{ab}}$,
$y[4..4]=\texttt{b}$, and
$y[5..5]=\texttt{b}$.
We sort these strings lexicographically and compute the longest common prefix information (Steps 3-4). We find that $\ell=2$: the maximal longest common prefixes are $\texttt{\textcolor{red}{ba}}$ and $\texttt{\textcolor{blue}{ab}}$, denoting that \texttt{\textcolor{red}{aba}} and \texttt{\textcolor{blue}{bab}} are the longest palindromic-preserved common factors of odd length. In fact, algorithm {\sc lPalcf} outputs $2\ell-1=3$ as \texttt{aba} and \texttt{bab} are the longest palindromic-preserved common factors of any length.
\end{example}

\section{Final Remarks}

In this paper, we introduced a new family of string processing problems.
The goal is to compute factors common to a set of strings preserving a specific property and having maximal length. We showed linear-time algorithms for square-free, periodic, and  palindromic factors under three different settings.
We anticipate that our paradigm can be extended to other string properties or settings.

\section*{Acknowledgements}
We would like to acknowledge an anonymous reviewer of a previous version of this paper who suggested the second linear-time algorithm for computing the longest periodic-preserved common factor.
Solon P.~Pissis and Giovanna Rosone are partially supported by the Royal Society project IE 161274 ``Processing uncertain sequences: combinatorics and applications''.
Giovanna Rosone and Nadia Pisanti are partially supported by the project Italian MIUR-SIR CMACBioSeq (``Combinatorial methods for analysis and compression of biological sequences'') grant n.~RBSI146R5L.

\bibliographystyle{plain}
\bibliography{references.bib}

\end{document}